\documentclass{article}

\usepackage{times}  
\usepackage{helvet}  
\usepackage{courier}  
\usepackage[hyphens]{url}  
\usepackage{graphicx} 
\usepackage{caption} 
\usepackage[letterpaper, margin=1in]{geometry}

\usepackage[linesnumbered,ruled,vlined]{algorithm2e}
\usepackage{natbib}
\SetKwInput{KwInput}{Input}                
\SetKwInput{KwOutput}{Output} 
\SetKwInput{KwPre}{Precondition}
\SetKwInput{KwInv}{Invariant}
\usepackage{hyperref}
\usepackage{amsmath,amsfonts,amssymb,bbm,amsthm,bm,xspace}
\usepackage{color}
\usepackage{thm-restate}
\usepackage{authblk}

\SetKwRepeat{Do}{do}{while}
\newcommand{\USW}{\mathcal{USW}}
\theoremstyle{definition}
\newtheorem{definition}{Definition}[section]
\theoremstyle{plain}

\newtheorem{theorem}{Theorem}[section]
\newtheorem{lemma}[theorem]{Lemma}
\newtheorem{proposition}[theorem]{Proposition}

\title{EFX Allocations Exist for Binary Valuations\footnote{This Paper has been accepted by International Joint Conference on Theoretical Computer Science – Frontier of Algorithmic Wisdom (IJTCS-FAW 2023)}}
\author[1]{Xiaolin Bu}
\author[2]{Jiaxin Song}
\author[3]{Ziqi Yu}
\affil[1]{Shanghai Jiao Tong University, lin\underline{ }bu@sjtu.edu.cn}
\affil[2]{Shanghai Jiao Tong University, sjtu\underline{ }xiaosong@sjtu.edu.cn}
\affil[3]{Shanghai Jiao Tong University, yzq.lll@sjtu.edu.cn}
\date{}

\begin{document}

\maketitle

\begin{abstract}
We study the fair division problem and the existence of allocations satisfying the fairness criterion \emph{envy-freeness up to any item} (EFX).
The existence of EFX allocations is a major open problem in the fair division literature.
We consider \emph{binary valuations} where the marginal gain of the value by receiving an extra item is either $0$ or $1$.
\cite{babaioff2021fair} proved that EFX allocations always exist for binary and \emph{submodular} valuations.
In this paper, by using completely different techniques, we extend this existence result to general binary valuations that are not necessarily submodular, and we present a polynomial time algorithm for computing an EFX allocation.

\end{abstract}

\section{Introduction}
\emph{Fair division} studies how to allocate heterogeneous resources fairly among a set of agents.
It is a classical resource allocation problem that has been widely studied by mathematicians, economists, and computer scientists.
It has a wide applications including school choices~\citep{abdulkadirouglu2005new}, course allocations~\citep{budish2012multi}, paper review assignments~\citep{lian2018conference}, allocating computational resources~\citep{ghodsi2011dominant}, etc.
Traditional fair division literature considers resources that are \emph{infinitely divisible}.
The fair division problem for infinitely divisible resources is also called \emph{the cake-cutting problem}, which dated back to~\cite{Steinhaus48,Steinhaus49} and has been extensively studied thereafter~\citep{aumann2015efficiency,aumann2012computing,brams2012maxsum,caragiannis2012efficiency,cohler2011optimal,bei2012optimal,bei2017cake,tao2022existence,BU2023103904}.
Among those fairness notions, \emph{envy-freeness} (EF) is most commonly considered, which states that each agent values her own allocated share at least as much as the allocation of any other agent.
In other words, each agent does not envy any other agent in the allocation.

Recent research in fair division has been focusing more on allocating \emph{indivisible items}.
It is clear that absolute fairness such as envy-freeness may not be achievable for indivisible items.
For example, we may have fewer items than agents.
\citet{CKMP+19} proposed a notion that relaxes envy-freeness, called \emph{envy-freeness up to any item} (EFX), which requires that, for any pair of agents $i$ and $j$, $i$ does not envy $j$ after the removal of \emph{any} item from $j$'s bundle.
Despite a significant amount of effort (e.g.,~\citep{plaut2020almost,chaudhury2020efx,chaudhury2021little,caragiannis2019envy,babaioff2021fair,berger2022almost,akrami2022efx,feldman2023optimal}), the existence of EFX allocations is still one of the most fundamental open problems in the fair division literature.
The existence of an EFX allocation is only known for two agents~\citep{plaut2020almost} or three agents with more special valuations~\citep{chaudhury2020efx,akrami2022efx}.

For a general number of agents, the existence of EFX allocations is only known for \emph{binary} and \emph{submodular} valuations, also known as \emph{matroid-rank} valuations~\citep{babaioff2021fair}.
Under this setting, \cite{babaioff2021fair}~designed a mechanism that outputs an EFX allocation and maximizes the \emph{Nash social welfare} at the same time, where the Nash social welfare of an allocation is the product of all the agents' values.
Unfortunately, Babaioff et al.'s techniques cannot be extended to the setting with binary valuations that are not necessarily submodular, as it is possible that all Nash social welfare maximizing allocations can fail to be EFX.
For example, consider a fair division instance with $n$ agents and $m> n$ items, and the valuations of the $n$ agents are defined as follows:
\begin{itemize}
    \item for agent $1$, she has value $1$ if she receives at least $m-n+1$ items, and she has value $0$ otherwise;
    \item for each of the remaining $n-1$ agents, the value equals to the number of the items received.
\end{itemize}
It is easy to check that the valuations are binary in the instance above.
However, the allocation that maximizes the Nash social welfare must allocate $m-n+1$ items to agent $1$ and allocate one item to each of the remaining $n-1$ agents, as this is the only way to make the Nash social welfare nonzero.
It is easy to check that such an allocation cannot be EFX if $m$ is significantly larger than $n$.

As a result, studying the existence of EFX allocations for general binary valuations requires different techniques from~\cite{babaioff2021fair}.

\subsection{Our Results}
We prove that EFX allocations always exist for general binary valuations which may not be submodular.
In addition, we provide a polynomial-time algorithm to compute such an allocation.
Our technique is based on the envy-graph procedure by~\cite{chaudhury2021little}.
\cite{chaudhury2021little} proposed a pseudo-polynomial-time algorithm that computes a \emph{partial} EFX allocation such that the number of the unallocated items is at most $n-1$ and no one envies the unallocated bundle.
We show that, for binary valuations, we can always find a complete EFX allocation, and it can be done in a polynomial time.
In particular, binary valuations enable some extra update steps which can further allocate the remaining unallocated items while guaranteeing EFX property.

\subsection{Related Work}
The existence of EFX allocations in general is a fundamental open problem in fair division.
Many partial progresses have been made in the previous literature.
\cite{plaut2020almost} showed that EFX allocations exist if agents' valuations are identical. 
By a simple ``I-cut-you-choose'' protocol, this result implies EFX allocations always exist for two agents.
\cite{chaudhury2020efx} showed that EFX allocations exist for three agents if the valuations are additive (meaning that the bundle's value equals to the sum of the values to the individual items).
The existence of EFX allocations for three agents is extended to slightly more general valuation functions by~\cite{feldman2023optimal}.

For binary valuations, \cite{babaioff2021fair} showed that EFX allocations always exist if the valuation functions are submodular. More details about this work have already been discussed in the introduction section, and this is the work most relevant to our paper.

Since the existence of EFX allocations is a challenging open problem, many papers study partial EFX allocations.
\cite{chaudhury2021little} proposed a pseudo-polynomial-time algorithm that computes a partial EFX allocation such that the number of the unallocated items is at most $n-1$ and no one envies the unallocated bundle.
For additive valuations, \cite{caragiannis2019envy} showed that it is possible to obtain a partial EFX allocation that achieves $0.5$-approximation to the optimal Nash social welfare;
\cite{berger2022almost} showed that, for four agents, we can have a partial EFX allocation with at most one unallocated item.

\citet{budish2011combinatorial} proposed \emph{envy-freeness up to one item} (EF1), which is weaker than EFX. 
It requires that, for any pair of agents $i$ and $j$, \emph{there exists} an item from $j$'s allocated bundle whose hypothetical removal would keep agent $i$ from envying agent $j$.
It is well-known that EF1 allocations always exist and can be computed in polynomial time \citep{Lipton04onapproximately,budish2011combinatorial}.

Other than fairness, previous work has also extensively studied several efficiency notions including social welfare, Nash social welfare, Pareto-optimality, and their compatibility with fairness~\citep{BKV18,CKMP+19,garg2021fair,aziz2020computing,barman2020optimal,bu2022complexity,bei2021price,caragiannis2019envy,chaudhury2020efx}.
We will not further elaborate on it here.

\section{Preliminaries}
Let $N$ be the set of $n$ agents and $M$ be the set of $m$ indivisible items.
Each agent $i$ has a valuation function $v_i: 2^M\rightarrow \mathbb{R}_{\ge 0}$ to specify her utility to a bundle, and we abuse the notation to denote $v_i(\{g\})$ by $v_i(g)$.
Each valuation function $v_i$ is \emph{normalized} and \emph{monotone}:
\begin{itemize}
    \item Normalized: $v_i(\emptyset)=0$;
    \item Monotone: $v_i(S)\geq v_i(T)$ whenever $T\subseteq S$.
\end{itemize}
We say a valuation function $v_i$ is \emph{binary} if $v_i(\emptyset) = 0$ and $v_i(S\cup \{g\})-v_i(S)\in \{0, 1\}$ for any $S\subseteq M$ and $g\in M\setminus S$.
In this paper, we will consider exclusively binary valuation functions.
Note that we do not require the valuation function to be additive.

An allocation $\mathcal{A}=(A_1, A_2, \dots, A_n)$ is a partition of $M$, where $A_i$ is allocated to agent $i$. A partial allocation $\mathcal{A}=(A_1,A_2,\ldots,A_n,P)$ is also a partition of $M$, where $P$ is the set of unallocated items. Notice a partial allocation is an allocation if and only if $P=\emptyset$. 
It is said to be \emph{envy-free} if $v_i(A_i)\ge v_i(A_j)$ for any $i$ and $j$.
However, an envy-free allocation may not exist when allocating indivisible items, for example, when $m < n$.
In this paper, we consider a relaxation of envy-freeness, called \textit{envy-freeness up to any item} (EFX) \citep{CKMP+19}.

\begin{definition}[EFX]
An allocation $\mathcal{A} = (A_1, \dots, A_n)$ satisfies envy-freeness up to any item (EFX) if $v_i(A_i)\ge v_i(A_j\setminus \{g\})$ for any $i, j$ and any item $g\in A_j$.
\end{definition}

\begin{definition}[Strong Envy]
    Given a partial allocation $(A_1,\ldots,A_n,P)$, we say that $i$ \textbf{strongly envy} $j$ if $v_i(A_i) <v_i(A_j\setminus\{g\})$ for some $g\in A_j$. Notice that we can extend this definition to a complete allocation by setting $P=\emptyset$.
\end{definition}

It is clear that an allocation is EFX if and only if there is no strong envy between every pair of agents.

Given an allocation $\mathcal{A}$, its \emph{utilitarian social welfare}, or simply \emph{social welfare}, is defined by 
$$\USW(\mathcal{A})=\sum_{i=1}^n v_i(A_i).$$

\subsection{Envy-Graph}
\begin{definition}[Envy-Graph]
    Given a partial allocation $\mathcal{A}=\{A_1,A_2,\ldots,A_n,P\}$, the envy graph $G=(V,E)$ is defined as follows. Each vertex in $V$ is considered as an agent. For agents $i,j\in N$, $(i,j)\in E$ if and only if $i$ envies $j$, i.e. $v_i(A_i)<v_i(A_j)$.
\end{definition}
\begin{definition}[Cycle-Elimination]\label{def:cycleelimination}
    For a cycle $u_0\to u_1\to\cdots u_{k-1}\to u_0$, \textbf{the cycle-elimination procedure} is performed as follows. 
    Each agent $u_i$ receives the bundle from $u_{i+1}$ for $i\in\{0,1,\ldots,k-1\}$(indices are modulo $k$). This process terminates when there is no cycle in the envy-graph. 
\end{definition}
It can be verified that a cycle-elimination step will not break the EFX property, since all bundles remain unchanged, what we have done is just exchanging these bundles. 
Also, the social welfare will increase since every agent in the cycle gets the bundle which is of higher value than her previous one.
\begin{definition}[Source agent]
    An agent is called a \textbf{source agent} if the in-degree of her corresponding vertex in the envy-graph is $0$. In other words, an agent is a source agent if no one envies her bundle.
\end{definition} 

Note that in the following part, we add another type of edge: maximal envy edges into the envy-graph.
However, when referring to the source agent, we only consider the induced subgraph of the envy-graph without the maximal envy edges.

\subsection{Binary Valuations and Pre-Envy}
Recall that we focus on binary valuations in this paper.
When valuation functions are binary, there are some special properties compared with general valuation functions. In this case, the marginal gain of a single item is $0$ or $1$. 

\begin{proposition}
    In an EFX allocation, if $i$ envies $j$, then $v_i(A_j)-v_i(A_i)=1$.
\end{proposition}
\begin{proof}
    Since $i$ envies $j$, $v_i(A_j)-v_i(A_i)>0$. Also, suppose $v_i(A_j)-v_i(A_i)\geq 2$. Since we focus on the binary valuation profile, for any $g\in A_j$
    \[v_i(A_i)+2\leq v_i(A_j)\leq v_i(A_j\setminus \{g\})+1.\]
    However, by the definition of EFX, we have $v_i(A_i)\geq v_i(A_j\setminus\{g\})$ for any $g\in A_j$, which leads to a contradiction.
\end{proof}
\begin{proposition}\label{prop:1}
    Suppose $j$ does not envy $i$'s bundle, i.e., $v_j(A_j)\geq v_j(A_i)$. If $j$ envies $i$'s bundle after adding an item $g$ to it, then $v_j(A_j)= v_j(A_i)$
\end{proposition}
\begin{proof}
    If not so, then $v_j(A_j)\geq v_j(A_i)+1$. Then $v_j(A_i\cup\{g\})\leq V_j(A_i)+1\leq V_j(A_j)$, which contradicts to the fact that $j$ will envy $i$'s bundle after adding an item to it.
\end{proof}
By Proposition \ref{prop:1}, suppose a new edge $(j,i)$ appears in the envy graph after adding an item to $i$'s bundle. Then, if we do not add this item, we have $v_j(A_j)=v_j(A_i)$. To better illustrate this observation, we introduce the following notion and consider this new relationship in the envy graph. 
\begin{definition}[Pre-Envy]
    Suppose $i,j\in N$ are two agents. We say $j$ \textbf{pre-envies} $i$, if $v_j(A_j)=v_j(A_i)$. In the envy graph, $j$ pre-envies $i$ will be represented by $j\dashrightarrow i$.
\end{definition}
In the binary valuation profile, we can jointly consider envy and pre-envy. Let $G=(V,E,E')$, where $V$ is the set of all agents, $E$ is the set of all envy relationships and $E'$ are the pre-envy ones. If edges in $E\cup E'$ form a cycle, we can also use the cycle-elimination procedure.

As a remark, for a cycle with pre-envy edges only, it may still exist after applying the cycle-elimination procedure.
For the cycle with pre-envy edges, we do not always eliminate it.
Our algorithm only eliminates this type of cycles under some particular scenarios.
We will provide more details in the next section.

As another remark, our notion of \emph{pre-envy edge} is the same as the \emph{equality edge} in the paper~\citep{bei2021fair}.
We choose a different word in this paper as we will sometimes use ``pre-envy'' as a verb.

\section{Existence of EFX Allocations}
In this section, we prove our main result.

\begin{theorem}\label{thm:1}
For any binary valuation profile $(v_1,\ldots,v_n)$, there exists an EFX allocation, and it can be computed by a polynomial-time algorithm.
\end{theorem}

\subsection{The Main Algorithm} 
We describe our algorithm in Algorithm~\ref{alg:main}.

\begin{algorithm}[H]\label{alg:main}
  \KwOut{a complete and EFX allocation $(A_1,A_2,\ldots,A_n)$}
  Let $A_i=\emptyset,i\in N$, and $P=M$\;
  \While{$P\neq\emptyset$}{
    Find an arbitrary item $g\in P$\;
    Update the partial allocation by $\mathcal{A} \leftarrow \texttt{Update}(\mathcal{A},g)$ according to Algorithm~\ref{alg:update}\;
    Update the envy graph and perform the cycle-elimination procedure (Definition~\ref{def:cycleelimination})\;
  }
    \Return{$(A_1,\ldots, A_n)$}
  \caption{Computing an EFX allocation with binary marginal gain}
\end{algorithm}

The algorithm starts with a partial allocation $\mathcal{A}=(\emptyset,\ldots,\emptyset, M)$.
In each iteration, it will consider an unallocated item $g$ and invoke the update function \texttt{Update}$(\mathcal{A},g)$.
In particular, we attempt to allocate $g$ to a source agent by applying one of the two update rules $U_0$ and $U_1$.
After that, we update the envy graph and perform the cycle-elimination procedure to guarantee the existence of source agents.

\subsubsection{Rule $U_0$}
The updating Rule $U_0$ considers a special case, where there exists a source agent $i\in N$, such that no agent will strongly envy the bundle $A_i\cup\{g\}$.
In this case, $U_0$ just allocates $g$ to this source agent $i$, which does not destroy the property of EFX.
\subsubsection{Rule $U_1$} 
If $U_0$ fails, then for every source agent $i$, adding $g$ to $i$'s bundle will cause at least one strong envy. Suppose $\{s_1,\ldots,s_k\}$ is the set of source agents. By assumption, after adding $g$ to $s_1$'s bundle, some agents will strongly envy $s_1$.

Before describing the update rule $U_1$, we first define some notions which are used later.

\begin{definition}[Safe Bundle and Maximal Envious Agent]\label{def:safe}
    Suppose a partial allocation $\mathcal{A}=(A_1,A_2,\ldots,A_n,P)$, an agent $i$ and an item $g\in P$ satisfy that
    \begin{itemize}
        \item[(a)] $\mathcal{A}$ is a partial EFX allocation, and
        \item[(b)] adding $g$ to $A_i$ causes someone strongly envy $i$.
    \end{itemize}
    We say that $S\subseteq A_i\cup \{g\}$ is a \textbf{safe bundle} with respect to $\mathcal{A},i$ and $g$, if
    \begin{enumerate}
        \item There exists $j\in N$ such that $j$ envies $S$. 
        \item For each agent $j$ that envies the bundle $A_i\cup \{g\}$  (i.e., $v_{j}(A_i\cup \{g\}) > v_j(A_j)$), $j$ does not strongly envy $S$ (i.e., 
        for any item $s\in S$, $v_j(S\setminus \{s\}) \le v_j(A_j)$).
    \end{enumerate}
    \textbf{The maximal envious agent} is one of the agents who envies $S$. 
\end{definition}

Intuitively, a safe bundle $S$ is a \emph{minimal} subset of $A\cup\{g\}$ such that someone still envies $S$.
The minimal property guarantees that no one will strongly envy $S$.

We also remark that, according to our definition, the maximal envious agent of $\mathcal{A},i,g$ may be agent $i$ herself.

\begin{lemma}\label{lem:1}
    For every triple $(\mathcal{A},i,g)$ satisfying (a) and (b) in Definition~\ref{def:safe}, a safe bundle $S\subseteq A_i\cup\{g\}$ and the corresponding maximal envious agent $a$ always exists, and they can be found in a polynomial time.
\end{lemma}
We defer the proof of Lemma~\ref{lem:1} to Section~\ref{sec:proof}.

By Lemma~\ref{lem:1}, we can always find $S\subseteq A_{s_1}\cup\{g\}$ and the maximal envious agent $c_1\in N$, such that 
\begin{itemize}
    \item Agent $c_1$ envies $S$, i.e., $v_{c_1}(S)>v_{c_1}(A_{c_1})$.
    \item No agent strongly envies $S$.
\end{itemize}

Since agent $c_1$ does not envy $s_1$ before $g$ is added to $A_{s_1}$ (as $s_1$ is a source agent) and envies $s_1$ after the addition of $g$, according to Proposition~\ref{prop:1}, agent $c_1$ pre-envies $s_1$ in the allocation $\mathcal{A}$ (where $g$ has not been added yet).
We will use $c_1\dashrightarrow_gs_1$ to denote this special pre-envy edge where $c_1$ is a maximal envious agent for $\mathcal{A},s_1,g$.
We say that $\dashrightarrow_g$ is a maximal envy edge.

Since we have assumed adding $g$ to each source agent causes strong envy, by Lemma~\ref{lem:1}, for each source agent $s_i\in\{s_1,\ldots,s_k\}$, there exists an agent $c_i$ such that $c_i\dashrightarrow_gs_i$.
It is possible that $s_i\dashrightarrow_g s_i$.
In this case, we add a self-loop to the envy graph.


In the lemma below, we show that, if the pre-condition of $U_0$ fails, then there must exist a cycle that only consists of edges in the original $G$ and the maximal envy edges.

\begin{lemma}\label{lem:precycle}
    For a partial EFX allocation $\mathcal{A}$ and an unallocated item $g$, if $A_i\cup\{g\}$ is strongly envied by someone for every source agent $i$, there must exist a cycle containing at least one source agent that only consists of edges in the original $G$ and the maximal envy edges. Moreover, the cycle can be found in polynomial time. 
\end{lemma}
\begin{proof}
    We have shown that, for any $s_i$, there exists $c_i$ such that $c_i\dashrightarrow_g s_i$.
    We start from $s_1$, and we find $c_1$ such that $c_1\dashrightarrow_g s_1$.
    We find the source $s_2$ from which $c_1$ is reachable, and we find $c_2\dashrightarrow_g s_2$.
    We keep doing this: whenever we are at $c_i$, we find the source $s_{i+1}$ from which $c_i$ is reachable; and we find $c_{i+1}$ such that $c_{i+1}\dashrightarrow_g s_{i+1}$.
    Notice that we are traveling backward along a path that consists only of graph edges and maximal envy edges.
    We can keep traveling until we find a source that has been already visited before, in which case we have found a cycle. 
    
    It is also easy to check that the above procedure can be done in polynomial time.
\end{proof}

Next, we find an arbitrary source agent $s$ in the cycle described in Lemma~\ref{lem:precycle}.
We apply Lemma~\ref{lem:1} and find a safe bundle $S\subseteq A_s\cup\{g\}$.
We replace $A_s$ with $S$, and perform an operation that is similar to cycle-elimination: let each agent in the cycle receives the bundle of the next agent.
In particular, the agent preceding $s$ is the one that maximally envies $s$; she will receive $S$ and get a higher value.

The lemma below justifies the correctness of $U_1$.

\begin{lemma}\label{lem:remainEFX}
    After applying the update rule $U_1$, the partial allocation remains EFX.
\end{lemma}
\begin{proof}
    After the update, the bundle $A_s$ is replaced by $S$.
    Definition~\ref{def:safe} implies that no one strongly envies $S$ and the agent in the cycle preceding $s$ receives a higher value by getting $S$.
    Each remaining agent in the cycle receives a weakly higher value.
    Thus, the resultant allocation is EFX after the update for the same reason that the cycle-elimination procedure does not destroy the EFX property.
\end{proof}

\begin{algorithm}[H]\label{alg:update}
  \caption{The update rules for allocation $\mathcal{A}$ and item $g\in P$}
  \SetKwProg{Fn}{Function}{}{end}
  \Fn{\emph{\texttt{Update}}(allocation $\mathcal{A}=(A_1,\ldots,A_n,P)$, item $g\in P$)}{
    \If{there exists a source agent $i$ such that adding $g$ to $A_i$ does not break the EFX property} {
        Perform $U_0(\mathcal{A},g)$\;
    }\Else{
        Perform $U_1(\mathcal{A},g)$\;
    }
  }
  \Fn{$U_0$($\mathcal{A}$, $g$)}{
  Allocate $g$ to $i$: $A_i\leftarrow A_i\cup \{g\}$ and update pool: $P\leftarrow P\setminus \{g\}$\;
  }
  \Fn{$U_1$($\mathcal{A}$, $g$)}{
     Find a cycle $C$ described in Lemma~\ref{lem:precycle}\;
     For a source agent $s$ on $C$, find the safe bundle $S$ for $(\mathcal{A},s,g)$ by Lemma~\ref{lem:1}\;
     Let $A_s\leftarrow S$ and $P\leftarrow P\cup A_s\cup\{g\}\setminus S$\;
     Run the cycle-elimination procedure for the agents on $C$\;
  }
\end{algorithm}

\subsection{Proof of Theorem \ref{thm:1}} \label{sec:proof}
We first prove the fact that a safe bundle always exists and can be calculated in polynomial time.
Notice that the updating rule $U_1$ makes sense if and only if Lemma \ref{lem:1} holds.

\noindent\textbf{Proof of Lemma \ref{lem:1}:}
Suppose $SE:=\{a_1,a_2,\ldots,a_k\}$ is the set of agents who strongly envy $i$ after adding $g$ to $A_i$.
In other words, they strongly envy $A_i\cup\{g\}$.
Firstly, we claim that there exists a bundle $S\subseteq A_i\cup\{g\}$ such that 
    \[v_{a_1}(S)>v_{a_1}(A_{a_1})\geq v_{a_1}(S\setminus \{g'\})\quad \text{for any }g'\in S.\]
In other words, $a_1$ envies $S$ but does not strongly envy $S$.

The set $S$ can be found by letting agent $a_1$ iteratively remove an item from $A_i\cup\{g\}$ until agent $a_1$ does not strongly envy the bundle.
If $S$ does not exist, it must be that 1) agent $a_1$ strongly envies some $S'\subseteq A_i\cup\{g\}$, and 2) removing any single item $g'$ from $S'$ will cause $a_1$ no longer envy this bundle, i.e.,
    \[v_{a_1}(S'\setminus \{g'\})\leq v_{a_1}(A_{a_1}). \]
However, this contradicts to the fact that $a_1$ strongly envies $S'$. Hence, such $S$ always exists.

After iteratively updating $S$, we remove the agents that do not strongly envy $S$ from $SE$.
Then, we maintain that $SE$ only contains the agents who strongly envy $S$.
For example, by the above procedure, $a_1$ is no longer in $SE$.
If $SE$ is empty now, then $S$ is the safe bundle and $a_1$ is the corresponding maximal envious agent. If $SE$ is not empty, then choose $a'\in SE$ and do the same procedures above. 
After this step, $S$ and $SE$ will be further reduced and $a'$ is no longer in $SE$.
Since the size of $SE$ is decreased at least by 1 in each round, there exists an agent $a_j\in SE$ who is the last one to remove items from the bundle. 
Suppose $S\subseteq A_i\cup\{g\}$ to be the final version of the bundle after item removals. 
Then $S$ is the safe bundle and $a_j\in SE$ is the corresponding maximal envious agent. The following algorithm illustrates this process. Here we use $S,a$ to store the latest version of the bundle and the last agent who remove items, respectively.

\begin{algorithm}[H]
    \caption{Computing safe bundle and corresponding the maximal envious agent}
    \KwIn{A partial allocation $\mathcal{A}$, $i$, $g$}
    \KwOut{a safe bundle $S$ and a maximal envious agent $a$}
    $S\leftarrow A_i\cup\{g\}$\;
    let $SE$ be the set of agents who strongly envies $S$\;
    \While{$SE$ is non-empty}{
        Find an index $j$, such that $a_j\in SE$\;
        While $a_j$ strongly envies $S$, let $a_j$ remove an item from $S$ such that $v_{a_j}(A_{a_j})<v_{a_j}(S)$\;
        If some items in $S$ have been removed in the process above, update $a\leftarrow a_j$\;
        Remove $a_j$ from $SE$\;
    }
    \Return{$S,a$}
\end{algorithm}

Notice that this algorithm will perform at most $n$ rounds, as at least one agent is removed from $SE$ in each iteration. Each round obviously costs a polynomial time. Then the safe bundle and the maximal envious agent can be found in polynomial time. \qed

\noindent\textbf{Proof of Theorem \ref{thm:1}}: 
Firstly, the social welfare is increased by at least $1$ after each application of $U_1(\mathcal{A},g)$ in \texttt{update}$(\mathcal{A},g)$:
for $U_1$, Definition~\ref{def:safe} and Lemma~\ref{lem:1} ensure that the agent receiving the safe bundle $S$ gets a strictly higher value than before. 
Secondly, $U_0$ can be applied for at most $m$ times between two applications of $U_1$, as there can be at most $m$ items in $P$ and $U_0$ decreases $|P|$ by exactly $1$.
Since the social welfare can be at most $mn$, the algorithm terminates in at most $mn$ applications of $U_1$. 
Therefore, the algorithm terminates with at most $m^2n$ applications of $U_0$ or $U_1$, which runs in a polynomial time.

Also, the EFX property is preserved after each round. $U_0$ clearly preserves the EFX property by its pre-condition. 
Lemma~\ref{lem:remainEFX} guarantees that the (partial) allocation remains EFX after applying $U_1$.

Hence, we conclude Theorem~\ref{thm:1} that an EFX allocation always exists and it can be found in polynomial time. \qed

\section{Conclusion}
In this paper, we studied the existence of an EFX allocation under a binary valuation profile. In particular, we proved that such an EFX allocation always exists and proposed a polynomial-time algorithm to compute it. Compared with the general valuation, ``pre-envy'' makes sense in binary valuation profiles and gives us extra properties to cope with EFX.

\bibliographystyle{plainnat}
\bibliography{reference}

\end{document}